\documentclass[10pt,a4paper]{article}
\newcommand{\longv}[1]{#1}
\newcommand{\shortv}[1]{}
\usepackage{color}
\usepackage{amsmath}
\usepackage{amsfonts,amssymb}
\usepackage{hyperref}
\usepackage{proof}
\usepackage{bussproofs}
\usepackage{stmaryrd}

\usepackage{amsthm}
\usepackage{tikz}
\usetikzlibrary{calc,arrows,intersections}
\usetikzlibrary{shapes,fit}
\usetikzlibrary{positioning}
\usetikzlibrary{patterns}
\usepackage{subcaption}
\usepackage{sansmath}
\usepackage{wrapfig}
\longv{\usepackage{a4wide}}
\longv{\usepackage{mathrsfs}}

\newcommand{\hide}[1]{}

\newcommand{\Odd}[1]{\mathsf{Odd}(#1)}
\newcommand{\Even}[1]{\mathsf{Even}(#1)}
\newcommand{\bnf}{::=}

\newcommand{\dom}[1]{\mathit{Dom}(#1)}

\newcommand{\distrs}[1]{\Delta({#1})}
\newcommand{\distrone}{\mathcal{D}}

\newcommand{\tensor}{\otimes}
\newcommand{\arrow}{\multimap}

\newcommand{\midd}{\; \; \mbox{\Large{$\mid$}}\;\;}

\newcommand{\gameone}{G}
\newcommand{\gametwo}{H}
\newcommand{\gamethree}{K}
\newcommand{\BB}{\mathbf{B}}
\newcommand{\Str}{\mathbf{S}}
\newcommand{\TStr}{\mathbf{T}}
\newcommand{\LTStr}{\mathbf{L}}
\newcommand{\opponent}[1]{O_{#1}}
\newcommand{\player}[1]{P_{#1}}
\newcommand{\legals}[1]{L_{#1}}

\newcommand{\legalspar}[2]{L_{#1}^{#2}}

\newcommand{\stratone}{\mathsf{f}}
\newcommand{\strattwo}{\mathsf{g}}

\newcommand{\stracopy}{\mathsf{once}}
\newcommand{\stracat}{\mathsf{mult}}
\newcommand{\strinfvar}{\mathsf{infvar}}
\newcommand{\strcoll}{\mathsf{coll}}
\newcommand{\strarandsof}{\mathsf{randsof}}

\newcommand{\movone}{x}
\newcommand{\movtwo}{y}
\newcommand{\playone}{s}

\newcommand{\alts}[2]{\text{Alt}(#1,#2)}

\newcommand{\gameparone}{G^{\circ}}
\newcommand{\Strpar}[1]{\Str[#1]}
\newcommand{\TStrpar}[1]{\TStr[#1]}
\newcommand{\LTStrpar}[1]{\LTStr[#1]}

\newcommand{\expar}[2]{!_{#1}#2}

\newcommand{\sof}[3]{\mathit{SOF}_{#1,#2,#3}}
\newcommand{\linsof}[2]{\mathit{LINSOF}_{#1,#2}}
\newcommand{\logsof}[1]{\mathit{LOGSOF}_{#1}}

\newcommand{\conv}[2]{#1\Downarrow^{#2} 1}
\newcommand{\pconv}[3]{#1\Downarrow^{#2} #3}

\newcommand{\probobs}[2]{\text{Prob}_{#1}(#2)}

\newcommand{\obpgames}{\mathbf{PolyPG^{\oplus}}}

\newcommand{\PG}{\mathscr{PG}}
\newcommand{\PolyPG}{\mathscr{PPG}}

\newcommand{\pols}{\mathbf{Pol}}

\newcommand{\idpol}{\iota}
\newcommand{\ulg}{\lceil\lg\rceil}
\newcommand{\moves}[1]{M_{#1}}
\newcommand{\restr}[3]{#1 \upharpoonright_{#2,#3}}
\newcommand{\caracplay}[1]{\overline{#1}}

\newcommand{\funone}{g}
\newcommand{\unif}[1]{U_{#1}}
\newcommand{\bfunone}{F}
\newcommand{\parone}{N} 
\newcommand{\partwo}{M}
\newcommand{\parthree}{L}
\newcommand{\nsetone}{S}

\newcommand{\Var}[2]{\mathbf{Var}_{#1}(#2)} 
\newcommand{\ev}[1]{\mathbb{E}(#1)} 
\newcommand{\Infl}[3]{\mathbf{Inf}_{#1}^{#2}(#3)}
\newcommand{\dtone}{T} 
\newcommand{\enq}[3]{\Delta_{#1,#2}(#3)}
\newcommand{\subst}[3]{#1[#2\leftarrow #3]}

\newcommand{\BinStr}{\mathbb{S}}
\newcommand{\PBinStr}[1]{\mathbb{S}[#1]}
\newcommand{\Bool}{\mathbb{B}}
\newcommand{\NN}{\mathbb{N}}
\newcommand{\RR}{\mathbb{R}}

\shortv{
\newenvironment{proof}{\begin{trivlist}
       \item[\hskip \labelsep {\bfseries Proof.}]}{\hfill $\Box$ \end{trivlist}}
}

\longv{
  \newtheorem{theorem}{Theorem}
  \newtheorem{definition}[theorem]{Definition}
  \newtheorem{lemma}[theorem]{Lemma}
  \newtheorem{proposition}[theorem]{Proposition}
  \newtheorem{example}[theorem]{Example}
}

\usepackage{xcolor}
\definecolor[named]{ACMBlue}{cmyk}{1,0.1,0,0.1}
\definecolor[named]{ACMYellow}{cmyk}{0,0.16,1,0}
\definecolor[named]{ACMOrange}{cmyk}{0,0.42,1,0.01}
\definecolor[named]{ACMRed}{cmyk}{0,0.90,0.86,0}
\definecolor[named]{ACMLightBlue}{cmyk}{0.49,0.01,0,0}
\definecolor[named]{ACMGreen}{cmyk}{0.20,0,1,0.19}
\definecolor[named]{ACMPurple}{cmyk}{0.55,1,0,0.15}
\definecolor[named]{ACMDarkBlue}{cmyk}{1,0.58,0,0.21}

\title{On Higher-Order Cryptography\\ (Long Version)}
\date{}
\shortv{
\titlerunning{On Higher-Order Cryptography}
\author{Boaz Barak}{Harvard University}{}{}{}
\author{Rapha\"elle Crubill\'e}{IMDEA Software Institute}{}{}{}
\author{Ugo Dal Lago}{University of Bologna \& INRIA}{}{}{}
\authorrunning{B. Barak et al.}

\authorrunning{Barak, Crubill\'e, and Dal Lago}

\Copyright{B. Barak, R. Crubill\'e and U. Dal Lago}



\hideLIPIcs
}
\longv{
  \author{Boaz Barak \and Rapha\"elle Crubill\'e \and Ugo Dal Lago}
  }

\begin{document}
\maketitle         
\begin{abstract}
  Type-two constructions abound in cryptography: adversaries for
  encryption and authentication schemes, if active, are modelled as
  algorithms having access to oracles, i.e. as second-order
  algorithms. But how about making cryptographic schemes
  \emph{themselves} higher-order? This paper gives an answer to this
  question, by first describing why higher-order cryptography is
  interesting as an object of study, then showing how the concept of
  probabilistic polynomial time algorithm can be generalised so as to
  encompass algorithms of order strictly higher than two, and finally
  proving some positive and negative results about the existence of
  higher-order cryptographic primitives, namely authentication schemes
  and pseudorandom functions.
\end{abstract}

\newcommand{\UGO}[1]{\textcolor{red}{#1}}
\newcommand{\BOAZ}[1]{\textcolor{blue}{#1}}
\newcommand{\RAPHA}[1]{\textcolor{violet}{#1}}

\section{Introduction}
Higher-order computation generalizes classic first-order one by
allowing algorithms and functions to not only take \emph{strings} but
also \emph{functions} in input.  It is well-known that this way of
computing gives rise to an interesting computability and complexity
theory~\cite{LongleyNormann,KawamuraCook12,Weihrauch13}, and
that it also constitutes a conceptual basis for the functional programming paradigm,
in which higher-order subroutines allow for a greater degree of modularity and
conciseness in programs.

In cryptography~\cite{GoldreichI06,GoldreichII06,KatzLindell},
computation is necessarily randomized, and being able to restrict the
time complexity of adversaries is itself crucial: most modern
cryptographic schemes are unsecure against computationally unbounded
adversaries. Noticeably, higher-order constructions are often
considered in cryptography, in particular when modeling active
adversaries, which have access to oracles for the underlying
encryption, decription, or authentication functions, and can thus
naturally be seen as second-order algorithms.

Another example of useful cryptographic constructions which can be
spelled out at different type orders, are \emph{pseudorandom}
primitives. Indeed, pseudorandomness can be formulated on (families
of) strings, giving rise to so-called pseudorandom
\emph{generators}~\cite{BlumMicali84}, but also on (families of)
first-order functions, giving rise to the so-called pseudorandom
\emph{functions}~\cite{GGM86}. In the former case, adversaries
(i.e., distinguishers) are ordinary polynomial time algorithms, while
in the latter case, they are polytime oracle machines.

Given the above, it is natural to wonder whether standard primitives
like encryption, authentication, hash functions, or pseudorandom
functions, could be made higher-order. As discussed in Section
\ref{sect:why} below, that would represent a way of dealing with 
code-manipulating programs and their security in a novel, fundamentally
interactive, way. Before even looking at the feasibility of this goal, there is one
obstacle we are facing, which is genuinely definitional: how could we
even \emph{give} notions of security (e.g. pseudorandomness, unforgeability,
and the like) for \emph{second-order} functions, given that those definitions
would rely on a notion of \emph{third-order} probabilistic polynomial time
adversary, itself absent from the literature? Indeed, although
different proposals exist for classes of feasible \emph{deterministic}
functionals~\cite{CookKapron89,KawamuraCook12}, not much is known if the underlying algorithm
has access to a source of randomness.
Moreover, the notion of feasibility used in cryptography is based
on the so-called \emph{security parameter}, a global numerical parameter which
controls the complexity of all the involved parties. In Section
\ref{sect:hocomplexity}, we give a definition of higher-order
probabilistic polynomial time by way of concepts borrowed from game
semantics~\cite{HO00,AJM00}, and being inspired by recent work by Fer\'ee~\cite{Feree17}. We give
evidence to the fact that the provided definition is general enough
to capture a broad class of adversaries of order strictly higher
than two.

After having introduced the model, we take a look at whether any
concrete instance of a secure higher-order cryptographic primitive can
be given. The results we provide are about pseudorandomness functions and
their application to authentication. We prove on the one hand that
those constructions are \emph{not} possible, at least if one insists
on them to have the expected type (see
Section~\ref{sect:compression}). On the other hand, we prove (in
Section~\ref{sect:positive} below) that second-order pseduorandomness \emph{is}
possible if the argument function takes as input a string of
\emph{logarithmic} length.

\section{The Why and How of Encrypting or Authenticating Functions}\label{sect:why}
Encryption and authentication, arguably the two simplest cryptographic
primitives, are often applied to \emph{programs} rather than
mere data.  But when this is done, programs are treated as ordinary data,
i.e., as \emph{strings of symbols}.  In particular, two different but equivalent
programs are seen as different strings, and their encryptions or
authentication tags can be completely different objects.
It is natural to ask the following: would it be possible to deal with
programs as \emph{functions} and not as \emph{strings}, in a
cryptographic scenario? Could we, e.g., encrypt or authenticate programs
seeing them as \emph{black boxes}, thus without any access to their
code?

For the sake of simplicity, suppose that the program $\texttt{P}$ we
deal with has a very simple IO behaviour, i.e. it takes an input a
binary string of length $n$ and returns a boolean. Authenticating it
could in principle be done by querying $\texttt{P}$ on some of its
inputs and, based on the outputs to the queries, compute a tag for
$\texttt{P}$. As usual, such an authenticating scheme would be secure
if no efficient adversary $\mathcal{A}$ could produce a tag
for $\texttt{P}$ without knowing the underlying secret key $k$ (such
that $|k|=n$) with non-negligible probability. Please notice that the adversary, contrarily to the
scheme itself, will have access to the code of
$\texttt{P}$, even if that code has not been used during the
authenticating process. 

But how could security can be \emph{defined} in a setting like the
above? The three entities at hand have the
following types, where $\mathtt{MAC}$ is the authentication
algorithm being defined, $\BinStr=\{0,1\}^*$ is the set of binary strings,
and $\Bool=\{0,1\}$ is the set of boolean values:
$$
\texttt{P}:\BinStr\rightarrow\Bool
\qquad
\mathtt{MAC}:\BinStr\rightarrow(\BinStr\rightarrow\Bool)\rightarrow\BinStr
\qquad
\mathcal{A}:((\BinStr\rightarrow\Bool)\rightarrow\BinStr)\rightarrow(\BinStr \rightarrow \Bool) \times\BinStr
$$
The first argument of $\mathtt{MAC}$ is the key $k$, which is of course
not passed to the adversary $\mathcal{A}$. The latter can query
$\mathtt{MAC}_k$ and produce a tagged message, and whose type, as expected, is
third-order. The above is not an accurate description of the
input-output behaviour of the involved algorithm, and in particular of
the fact that length of the input string to $\texttt{P}$ might be in a
certain relation to the length of $k$, i.e., the underlying
security parametery. Reflecting all this in the types is however possible by replacing occurrences
of the type $\BinStr$ with \emph{refinements} of it, as follows:
\begin{align*}
\texttt{P}&:\PBinStr{n}\rightarrow\Bool\\
\mathtt{MAC}&:\PBinStr{n}\rightarrow(\PBinStr{r(n)}\rightarrow\Bool)\rightarrow\PBinStr{p(n)} \\
\mathcal{A}&:((\PBinStr{r(n)}\rightarrow\Bool)\rightarrow\PBinStr{p(n)})\rightarrow (\PBinStr{r(n)} \rightarrow \Bool )\times\PBinStr{p(n)}
\end{align*}
But how could the time complexity of the
three algorithms above be defined?  While polynomial time computability for the
function $\mathtt{P}$ and the authentication $\mathtt{MAC}$ can be
captured in a standard way using, e.g., oracle Turing machines, the
same cannot be said about $\mathcal{A}$. How to, e.g., appropriately account for the
time $\mathcal{A}$ needs to ``cook'' a function $f$ in $\PBinStr{n}\rightarrow\Bool$
to be passed to its argument functional? So, appealing as it is, our objective
of studying higher-order forms of cryptography is bound to be nontrivial,
even from a purely \emph{definitional} point of view.

The contributions of this paper can be now described at a more
refined level of detail, as follows:
\begin{itemize}
\item
  On the one hand, we give a \emph{definition} of a polynomial-time
  higher-order probabilistic algorithm whose time complexity depends on
  a global security parameter and which is captured by games and
  strategies, in line with game
  semantics~\cite{HO00,AJM00}. This allows to discriminate satisfactorily
  between efficient and non-efficient adversaries, and accounts for
  the complexity of first-and-second-order algorithms consistently
  with standard complexity theory. 
\item
  On the other hand, we give some positive and negative results about the possibility
  of designing second-order cryptographic primitives, and in particular
  pseudorandom functions and authentication schemes. In particular we prove, by
  an essentially information-theoretic argument, that secure deterministic second-order
  authentication schemes cannot exist. A simple
  and direct reduction argument shows that a more restricted form of
  pseudorandom function exists under standard cryptographic assumptions.
  Noticeably, the adversaries we prove the existence of are of a very restricted
  form, while the ones which we prove impossible to build are quite general.
\end{itemize}
\section{Higher-Order Probabilistic Polynomial Time Through Parametrized Games}\label{sect:hocomplexity}
In this section, we introduce a framework inspired by game semantics
in which one can talk about the efficiency of probabilistic
higher-order programs in presence of a global security
parameter. While the capability of interpreting higher-order programs
is a well-established feature of game semantics, dealing \emph{at the
  same time} with probabilistic behavior \emph{and} efficiency constraints
has---to the best of the authors' knowledge---not been considered so
far. The two aspects have however been tackled
\emph{independently}. Several game models of probabilistic languages
have been introduced: we can cite here for instance the fully abstract
model of probabilistic Idealized Algol by Danos and
Harmer~\cite{DanosHarmer02}, or the model of probabilistic
\textsf{PCF} by Clairambault at al.~\cite{CP18}. About efficiency, we
can cite the work by F\'er\'ee~\cite{Feree17} on higher-order
complexity and game semantics, in which the cost of evaluating
higher-order programs is measured parametrically on the size of
\emph{all} their inputs, including functions, thus in line with
type-two complexity~\cite{CookKapron89}. We are instead interested in
the efficiency of higher-order definitions with respect to the
\emph{security parameter}. Unfortunately, existing probabilistic game
models do not behave well when restricted to capture feasibility:
\emph{polytime computable} probabilistic strategies in the spirit of
Danos and Harmer do not \emph{compose} (see the Extended Version of
this paper for more details).

Contrary to most works on game
semantics, we do not aim at building a model of a \emph{particular}
programming language, but we take game semantics as our model of
compuation. As a consequence, we are not bound by requirements to
model particular programming features or to reflect their
discriminating power.

We present our game-based model of computation in three steps: first
we define a category of deterministic games and strategies called
$\PG$---for \emph{parametrized games}---which capture computational
agents whose behavior is parametrized by the security parameter. This
model ensures that the agent are \emph{total}: they \emph{always}
answer any request by the opponent. In a second step, we introduce
$\PolyPG$, as a sub-category of $\PG$ designed to model those agents
whose time complexity is \emph{polynomially bounded} with respect to the security
parameter. Finally, we deal with \emph{randomized 
  agents} by allowing them to interact with a \emph{probabilistic
  oracle}, that outputs (a bounded amount of) random bits.
\subsection{Parametrized Deterministic Games}
Our game model has been designed so as to be
able to deal with security properties that---as exemplified by
\emph{computational indistinguishability}---are expressed by looking at the behavior
of adversaries \emph{at the limit}, i.e.,~when the security parameter
tends towards infinity. The \emph{agents} we consider
are actually \emph{families} of functions, indexed by the security
parameter. As such, our game model can be seen as a parametrized
version of Wolverson's simple games~\cite{Wolverson09}, where the set
of plays is replaced by a \emph{family} of sets of plays, indexed by
the natural numbers. Moreover, we require the total length of any
interaction between the involved parties to be polynomially bounded
in the security parameter.

We need a few preliminary definitions before delving into
the . Given two sets $X$ and $Y$, we define $\alts{X}{Y}$ as
$
\{(a_1,\ldots, a_n) \mid a_{i} \in X \text{ if $i$ is odd},
a_{i} \in Y \text{ if $i$ is even} \}
$, i.e., as the set of finite alternating sequences whose first element
is in $X$. Given any set of sequences $X$, $\Odd{X}$ (respectively,
$\Even{X}$) stands for the subset of $X$ containing the sequences of
odd (respectively, even) length.  \newcommand{\len}[1]{|#1|} From now
on, we implicitly assume that any (opponent or player) move $m$ can be
faithfully encoded as a string in an appropriate, fixed alphabet.
This way, moves and plays implicitly have a length, that we will note
$\len{\cdot}$.  We fix a set $\pols$ of unary polynomially-bounded
total functions on the natural numbers, which includes the identity
$\idpol$, base-$2$ logarithm $\ulg$, addition, multiplication,
and closed by composition. $\pols$ can be equipped with the pointwise
partial order: $p\leq q$ when $\forall n \in \NN, p(n) \leq q(n)$.
\begin{definition}[Parametrized Games]
A \emph{parametrized game} $\gameone = (\opponent \gameone, \player
\gameone, \legals \gameone)$ consists of sets $\opponent \gameone$,
$\player \gameone$ of opponent and player moves, respectively,
together with a family of non-empty prefix-closed sets $\legals
\gameone=\{\legalspar{\gameone}{n}\}_{n\in\NN}$, where
$\legalspar{\gameone}{n}\subseteq\alts{\opponent{\gameone}}{\player{\gameone}}$,
such that there is $p\in\pols$ with
$\forall n \in \NN.\forall s\in\legalspar{\gameone}{n}.\len{s}\leq p(n)$.
The union of $\opponent\gameone$ and $\player\gameone$ is indicated
as $\moves\gameone$.
\end{definition}
For every $n\in\NN$, $\legalspar{\gameone}{n}$ represents the set of legal plays,
when $n$ is the current value of the security parameter. Observe that the first move
is always played by \emph{the opponent}, and that for any fixed value of
the security parameter $n$, the length of legal plays is bounded by
$p(n)$, where $p\in\pols$.

\newcommand{\UnitGame}{\textbf{1}}
\newcommand{\BoolGame}{\mathbf{B}}
\newcommand{\OracleGame}[1]{\mathbf{O}^{#1}}
\begin{example}[Ground Games]\label{exa:groundgames}
  We present here some games designed to model \emph{data-types} in
  our computation model.  The simplest game is probably the \emph{unit
    game} $\UnitGame = (\{?\}, \{*\}, \{\legalspar{\UnitGame}{n}\}_{n
    \in \NN})$ with just one opponent move and one player move, where
  $\legalspar{\UnitGame}{n}=\{\varepsilon,?,?*\}$ for every $n$. Just
  slightly more complicated than the unit game is the \emph{boolean
    game} $\BoolGame$ in which the two moves $0$ and $1$ take the place
  of $*$.
  In the two games introduced so far, parameterisation is not really
  relevant, since
  $\legalspar{\gameparone}{n}=\legalspar{\gameparone}{m}$ for every
  $n,m\in\NN$.  The latter is not true in
  $\Strpar p = (\{?\}, \{0,1\}^\star, \{\legalspar{\Strpar p}{n}\}_{n
    \in \NN})$ with $\legalspar{\Strpar p}{n} = \{\varepsilon, ?\}
  \cup\{? s \mid \lvert s \rvert = p(n)\}$, which will be our way
  of capturing strings. A slight variation of $\Strpar{p}$ is
  $\LTStrpar{p}$ in which the returned string can have length
  \emph{smaller or} equal to $p(n)$.
\end{example}
\begin{example}[Oracle Games]
  As another example, we describe how to construct \emph{polynomial boolean
    oracles} as games. For every polynomial $p\in\pols$
  we define a game $\OracleGame{p}$ as $\OracleGame{p} =
  (\{?\},\{0,1\},\{\legalspar{\OracleGame{p}}{n}\}_{n \in \NN} ) $ with
  $$\legalspar
  {\OracleGame{p}}{n} = \{?\} \cup \{ ?b_1? b_2 \ldots ?b_m\mid b_i \in \{0,1\}\wedge m\leq p(n)\}
  \cup \{ ?b_1? b_2 \ldots ?b_m ? \mid b_i \in \{0,1\}\wedge m<p(n) \}.
  $$
\end{example}
Our oracle games are actually a special case of a more general
construction, that consists of taking any game $\gameone$, and a
polynomial $p$, and to build a game which consists in playing $p(n)$
times with the game $\gameone$.
\begin{definition}[Bounded Exponentials]
  Let $\gameone ={(\opponent \gameone, \player \gameone, \legals
    \gameone)} $ be a parametrized game. For every $p \in \pols$, we define
  a new parametrized game $\expar p \gameone :=(\opponent {\expar p \gameone},
  \player {\expar p \gameone}, \legals {\expar p \gameone})$
  as follows:
\begin{itemize}
\item
  $ \opponent {\expar p \gameone} = \NN \times \opponent \gameone$,
  and $\player {\expar p \gameone} = \NN \times \player \gameone$;
\item
  For $n \in \NN$, $\legalspar {\expar p \gameone} n$ is the set of those
  plays $\playone \in \alts{\opponent{\expar p \gameone} }{\player{\expar p
      \gameone}} $ such that:
  \begin{itemize}
  \item
    for every $i$, the $i$-th projection $\playone_i$ of $\playone$
    is in $\legalspar{\gameone}{n}$.
  \item
    if a move $(i+1,z)$ appears in $\playone$, then a move $(i,x)$
    appears at some earlier point of $\playone$, and $i+1 \leq p(n)$.
  \end{itemize}
\end{itemize}
\end{definition}
Observe that the game $\OracleGame{p}$ is isomorphic to the game $\expar p
\BoolGame$---we can build a bijection between legals plays having
the same length.

Games specify \emph{how} agents could play in a certain interactive
scenario. As such, they do not represent \emph{one} such agent, this
role being the one of strategies. Indeed, a strategy on a game is
precisely a way of specifying a \emph{deterministic} behavior of the player,
i.e.~how the player is going to answer to any possible behavior of an
adversary. We moreover ask our strategies to be total, in the sense
that the player cannot refuse to play when it is their turn.
\begin{definition}[Strategies]
A \emph{strategy} on a parametrized game
$\gameone = (\opponent \gameone,
\player \gameone, \legals \gameone)$
consists of a family
$\stratone=\{\stratone_n\}_{n\in\NN}$, where $\stratone_n$ is
a partial function from $\Odd{\legalspar \gameone n}$ to ${\player \gameone}$
such that:
\begin{itemize}
\item
  for every $\playone\in\Odd{\legalspar \gameone n}$, if
  $\stratone_n(\playone)=\movone$ is defined, then
  $\playone\movone\in\legalspar{\gameone}{n}$;
\item
  $\playone \movone \movtwo \in {\dom {\stratone_n}}$ implies that
  $\movone = {\stratone_n(\playone)}$;
\item
  $\playone \in \caracplay{\stratone}_n \, \wedge\, \playone \movone \in \legalspar \gameone n \Rightarrow \playone \movone \in \dom{\stratone_n} $.
\end{itemize}
If $\stratone$ is a strategy, we define the \emph{set of plays
  caracterising $\stratone$}, that we note $\caracplay \stratone$, as
the family $\{\caracplay{\stratone}_n\}$ where
$\caracplay{\stratone}_n = \{\varepsilon\} \cup \{\playone \stratone_n(\playone) \mid
\playone \in \dom \stratone \} \subseteq \legalspar{\gameone}{n}$.
\end{definition}
Any strategy $\stratone$ is entirely caracterized by its set
of plays $\caracplay \stratone$. Intuitively, a strategy stands for a
\emph{deterministic recipe} that informs the player on how to respond
to \emph{any} behavior of the opponent. As such, it does not need to
be effective. 

Up to now, the games we have described are such that their strategies
are meant to represent concrete data: think about how a strategy for, e.g.,
$\BoolGame$ or even $\OracleGame{p}$ could look like. It is now time to build
games whose strategies are \emph{functions}, this being the
following construction on games:
\begin{definition}[Constructing Functional Games]
The game $\gameone \arrow \gametwo$ is given as $\opponent{\gameone
  \arrow \gametwo} = \player \gameone + \opponent \gametwo$,
$\player{\gameone \arrow \gametwo} = \opponent \gameone + \player
\gametwo$, and $\legals{\gameone \arrow \gametwo} = \{\playone \in
\alts{\opponent{\gameone \arrow \gametwo}}{\player{\gameone \arrow
    \gametwo}} \mid \playone_0 \in \legals \gameone, \playone_1 \in
\legals \gametwo \} $.
\end{definition}
\begin{example}\label{ex:oraclestrategies}
  As an example, we look at the game $\OracleGame{p} \arrow \Strpar p$, which
  captures functions returning a string of size $p(n)$ after
  having queried a binary oracle at most $p(n)$ times. First, observe
  that:
  $$\OracleGame{p} \arrow \Strpar p = (\{?_{\Strpar p}, 0,1\}, \{?_{\OracleGame{p}}\}
  \cup \{0,1\}^\star, (\legalspar n {\OracleGame{p} \arrow \Strpar p})_{n \in
    \NN}),
  $$
  where $\legalspar n {\OracleGame{p} \arrow \Strpar p}$ is generated
  by the following grammar:
\begin{align*}
  & q \in \legalspar n {\OracleGame{p} \arrow \Strpar p} \bnf ?_{\Strpar p} \midd   ?_{\Strpar p} o \midd  ?_{\Strpar p}  e \midd  ?_{\Strpar p} e  s \qquad s \in \{0,1\}^{\star} \text{ with } \len s \leq p(n) \\
  & e \bnf \epsilon \midd ?_{\OracleGame{p}} b_1 \ldots ?_{\OracleGame{p}} b_m \qquad   o \bnf  ?_{\OracleGame{p}} \midd ?_{\OracleGame{p}} b_1 \ldots ?_{\OracleGame{p}} b_{m-1} ?_{\OracleGame{p}} \qquad b_i \in \{0,1\}, m \leq p(n)
\end{align*}
Of course there are \emph{many} strategies for this game, and
we just describe two of them here which both make use of the oracle: the
first one---that we will call $\stracopy$ asks
for a random boolean to the oracle, and when this boolean is $0$,
returns the string $0^n$, and returns the string $1^n$
otherwise. It is represented in Figure~\ref{ex:strat_s1}. The second
strategy---denoted $\stracat$, and represented in
Figure~\ref{ex:strat_s2}---generates a random key of
length $p(n)$ by making $p(n)$ calls to the probabilistic oracle.
\begin{figure}
  \centering
  \fbox{\begin{minipage}{.97\textwidth}
   \begin{center}
    \begin{subfigure}[b]{0.4\textwidth}
\begin{center}
  \begin{tikzpicture}[scale=0.6]
    \node (t1) at (1,1){$\OracleGame{p}$};
    \node (t2) at (3,1){$\arrow$};
    \node (t3) at (5,1){$\Strpar p$};
    \node (o1) at (0,0){O};
    \node (o2) at (0,-2){O};
    \node (mo1) at (5,0){$?_{\Strpar p}$};
     \node (mo2) at (1,-2){$b$};
     \node (p1) at (0,-1){P};
     \node (mp1) at (1,-1){$?_{\OracleGame{p}}$};
     \node (p2) at (0,-3){P};
      \node (mp2) at (5,-3){$b^n$};
  \end{tikzpicture}
\end{center}
  \caption{The strategy $\stracopy$.}\label{ex:strat_s1}
\end{subfigure}
    \begin{subfigure}[b]{0.4\textwidth}
\begin{center}
  \begin{tikzpicture}[scale=0.6]
    \node (t1) at (1,1){$\OracleGame{p}$};
    \node (t2) at (3,1){$\arrow$};
    \node (t3) at (5,1){$\Strpar p$};
    \node (o1) at (0,0){O};
    \node (mo1) at (5,0){$?_{\Strpar p}$};
     \node (p1) at (0,-1){P};
     \node (mp1) at (1,-1){$?_{\OracleGame{p}}$};
\node (o2) at (0,-2){O};
     \node (mo2) at (1,-2){$b_1$};

 \node at (0,-3){P};
     \node at (1,-3){$?_{\OracleGame{p}}$};
\node  at (0,-4){O};
     \node  at (1,-4){$b_2$};
      \node at (0,-5){$\vdots$};
     \node (p2) at (0,-6){P};
      \node (mp2) at (5,-6){$b_1 \ldots b_{p(n)}$};
  \end{tikzpicture}
\end{center}
\caption{The strategy $\stracat$}\label{ex:strat_s2}
    \end{subfigure}\end{center}
\end{minipage}}
    \caption{Two different strategies on the game $\OracleGame{p} \arrow \Strpar p$}
    \end{figure}
\end{example}
\longv{
\begin{example}
In this example, we look at the way of \emph{weakening} our
probabilistic oracle: more precisely, if we have an oracle that makes
available $p(n)$ booleans--for each value of the security parameter
$n$--then we can build from it an oracle that gives only $q(n)$
booleans, when $q \leq p$, by taking only the $q(n)$ first iteration
of our original oracle. We do this construction formally in a slighlty
more general setting: for any game $\gameone \in SG$, we build a
canonical strategy $w^\gameone_{q,p}$ on the game $!_p \gameone \arrow
!_q \gameone$.
\end{example}}

We now look at how to \emph{compose} strategies: given a strategy on $\gameone
\arrow\gametwo$, and $\gametwo\arrow\gamethree$, we want to build a strategy on $\gameone\arrow
\gamethree$ that combines them. We define composition as in~\cite{Wolverson09},
except that we need to take also in count the security parameter $n$.
\begin{definition}[Composition of Strategies]
  Let $\gameone,\gametwo,\gamethree$ be parametrized games, and let $\stratone$, $\strattwo$ be two
  strategies on $\gameone\arrow\gametwo$ and $\gametwo\arrow\gamethree$ respectively. We first
  define the set of \emph{interaction sequences of $\stratone$ and
    $\strattwo$} as:
  $$
  (\stratone \interleave \strattwo)_n = \{\playone \in (\moves \gameone +
  \moves \gametwo + \moves\gamethree)^\star \mid \restr \playone \gameone \gametwo\in
  \caracplay \stratone_n \wedge \restr \playone\gametwo\gamethree \in \caracplay \strattwo_n
  \},
  $$
  From there, we define the composition of $\stratone$ and
  $\strattwo$ as the unique strategy $\stratone;\strattwo$ such that:
  $$
  \caracplay{\stratone;\strattwo}_n = \{ \restr \playone\gameone\gamethree\mid \playone
  \in (\stratone \interleave \strattwo)_n\}.
  $$
\end{definition}
We can check that $\stratone;\strattwo$ is indeed a strategy on $A \arrow C$,
and that moreover the composition seen as an operation on strategies
is associative and admits an identity. We cat thus define $\PG$ as the category
whose objects are parametrized games, and whose set of morphisms
$\PG(\gameone,\gametwo)$ consists of the parametrized strategies on the game
$G \arrow H$.
\subsection{Polynomially Bounded Parametrized Games}

Parametrized games have been defined so as to have polynomially
bounded \emph{length}. However, there is no guarantee on the
effectiveness of strategies, i.e., that the next player move,
can be computed algorithmically from the history, uniformly in
the security parameter. This can be however tackled by considering
a subcategory of $\PG$ in which strategies are not merely
functions, but can be computed efficiently:
\begin{definition}[Polytime Computable Strategies]
  Let $\gameone$ a parametrized game in $\PG$, and $\stratone$ be a strategy on
  $\gameone$. We say that $\stratone$ is \emph{polytime computable} when
  there exists a polynomial time Turing machine which on the entry
  $(1^n, \playone)$ returns $\stratone(n)(\playone)$.
\end{definition}
All strategies we have given as examples in the previous
section are polytime computable. For example, the two strategies
from Example \ref{ex:oraclestrategies} are both computable in
linear time.
\begin{proposition}[Stability of Polytime Computable Strategies]
  Let $\gameone, \gametwo, \gamethree$ be polynomially bounded games.
  If $\stratone, \strattwo$ are polytime computable strategies,
  respectively on $\gameone \arrow \gametwo$, and $\gametwo \arrow
  \gamethree$, then $\stratone;\strattwo$ is itself a polytime
  computable strategy. 
\end{proposition}
We cat thus write $\PolyPG$ for the the sub-category of $\PG$ whose
objects are paramterized games, and whose morphisms are polytime
computable strategies.  \hide{
\subsection{Oracle Games}
We could define probabilistic strategies as follows:
\begin{itemize}
\item Games: $ \gameone =(\opponent \gameone,
  \player \gameone, (\legalspar \gameone n)_{n \in \NN}) $
\item Probabilistic Strategies: $\stratone : \NN \rightarrow (\legalspar \gameone n \cap \text{Odd} \rightarrow \distrs{\player \gameone} )$
\end{itemize}
Unfortunately, this is not stable by composition. Our solution:
\begin{definition}[The sub-category $\obpgames$]
\begin{itemize}
    \item Objects: parametrized \textbf{bounded} games
    \item Morphisms $\gameone \rightarrow \gametwo$: $(p,\stratone)$ with:
      \begin{itemize}
      \item $p \in \pols$
      \item $\stratone$ \textbf{polytime} computable on  $\expar p \BB \arrow (\gameone \arrow \gametwo)$ .
      \end{itemize}
\end{itemize}
\end{definition}
This forms a category with some nice constructions.
}

Let us now consider $\texttt{MAC}$ from Section~\ref{sect:why}. Its type
can be turned into the parametrized game
$\Strpar{\idpol}\arrow !_{q}(\Strpar{\idpol}\arrow\BoolGame)\arrow\Strpar{p}$.
The bounded exponential $!_{q}$ serves to model the fact that the argument
function can be accessed a number of times which is \emph{polynomially
  bounded} on $n$. As a consequence, $\texttt{MAC}$ can only query the argument
function a number of times which is negligibly smaller than the number of possible
queries, itself exponential in $n$. As we will see in Section~\ref{sect:impossibility},
this is the crux in proving security of such a message authentication code
to be unatteinable.
\subsection{Probabilistic Strategies}
Both in $\PG$ and in $\PolyPG$, strategies on any game $\gameone$ are purely
\emph{functions}, and the way they react to stimuli from the environment is
completely deterministic. How could we reconcile all this with our claim that
the framework we are introducing adequately model \emph{probabilistic} higher-order
computation? Actually, one could be tempted to define a notion of \emph{probabilistic
  strategy} in which the underlying function family $\{f_n\}_{n\in\NN}$ is
such that $f_n$ returns \emph{a probability distribution} $f_n(s)$ of player moves when fed
with the history $s$. This, however, would lead to some technical problems
when \emph{composing strategies}: it would be hard to keep the composition
of two efficient strategies itself efficient.

It turns out that a much more convenient route consists, instead, in
defining a \emph{probabilistic} strategy on $\gameone$ simply as a
\emph{deterministic} polytime strategy on
$\OracleGame{p}\arrow\gameone$, namely as an ordinary strategy having
access to polynomially many random bits.  Actually, we have already
encountered strategies of this kind, namely $\stracopy$ and $\stracat$
from Example~\ref{ex:oraclestrategies}. This will be our way of
modeling higher-order probabilistic computations.
\begin{definition}
  Let $\stratone$ be a strategy on
  the game $\OracleGame{p} \arrow \BoolGame$.
  For every $b\in\Bool$, we define the
  \emph{probability of observing $b$ when
    executing $\stratone$} as follows:
  $$
  \Pr(\pconv{\stratone}{n}{b})=
  \sum_{\stackrel{(b_1, \ldots, b_k) \in \Bool^k  }{\text { with } (?_{\BoolGame} \cdot ?_{\OracleGame{p}} \cdot b_1 \ldots ?_{\OracleGame{p}} \cdot b_k\cdot b)  \in (\caracplay \stratone)_n }} \frac{1}{2^k}
  $$
\end{definition}

Given a probabilistic strategy $\stratone$ on $\gameone$ (i.e. a strategy
on $\OracleGame{q}\arrow\gameone$) and $p\in\pols$, we indicate
as $!_{p}\stratone$ the strategy in which $p(n)$ copies of $\stratone$ are played,
\emph{but in which} randomness is resolved just once and for all, i.e. $!_{p}\stratone$
is a strategy on $\OracleGame{q}\arrow!_{p}\gameone$
  
\longv{
We want now to be able to \emph{observe} the result of a
computation. Similarly to what is done for instance in~\cite{Wolverson09}, we
observe only at ground types: for instance for the type Bool, the
observables will be a pair $(p_0,p_1)$, where $p_0$ is the probability
that the program returns false, and $p_1$ the probability that the
program returns true. In our model, the ground types are those games
where the only legal plays are the initial question by the opponent,
followed by a terminal answer of the player.

\begin{definition}
  Let $\gameone$ a game in $SG$. We say that $\gameone$ is \emph{observable} when for every $n \in \NN$:
  $$\legalspar \gameone n \subseteq \{?_\gameone\} \cup \{ ?_\gameone r \mid r \in \player \gameone \}. $$
  \end{definition}
Observe that the games $\BoolGame$ and $\Strpar p$ that we considered in Example~\ref{exa:groundgames} are indeed observable games. However, it is not the case of the game build using the $\arrow$ construct, for instance $\Strpar p \arrow \Strpar {2p}$.

Recall that in our framework, a \emph{probabilistic} program of some ground type $\gameone$ is actually represented by a strategy $\stratone$ on the game $\OracleGame{p} \arrow \gameone$, for some polynomials $p$. At this point, it should be noted that this polynomial $p$ is somehow arbitrary, since if we take a polynomial $q \geq p$, we can easily transform $\stratone$ in a strategy on $\OracleGame{q} \arrow \gameone$, using the strategy on $\OracleGame{q} \arrow \OracleGame{p}$ highlighted in Example~\ref{}. In order to recover information about the probabilistic behavior of the program, we compute the probability that the return value will be $r$: to do that we look at all the possible sequences of booleans outputed by the oracle that lead to the player returnin $r$ under the strategy $\stratone$, and then we suppose that those sequences are uniformly distributed, in a probabilistic sense. 
\begin{definition}
  Let $\gameone$ be an observable game, and $\stratone$ a strategy on the game $\OracleGame{p} \arrow \gameone$. For any $r \in \player \gameone$, we define the \emph{probability of observing $r$ when executing $\stratone$} as the function $\probobs \stratone r : \NN \rightarrow [0,1]$ such that
    $$
    n \mapsto \sum_{k \leq p(n)} \sum_{\stackrel{(b_1, \ldots, b_k) \in \{0,1\}^k  }{\text { with } (? \cdot ?_\BB \cdot b_1 \ldots ?_\BB \cdot b_k\cdot r)  \in (\caracplay \stratone)_n }} \left(\frac 1 2 \right)^k
   $$
  \end{definition}
Observe that if $\stratone$ is a strategy on $\OracleGame{p} \arrow \stratone$,
and $w_{q,p}$ the canonical strategy on $\OracleGame{q} \arrow \OracleGame{p}$--as defined
in Example~\ref{}, then for every value $r$, $\probobs \stratone r =
\probobs{w_{q,p}; \stratone} r$: it means that the probability of
observing $r$ is independant of the choice of the polynomial
oracle--as soon as it produces \emph{at least} as many random booleans
as asked by the agent.}
\section{The (In)feasibility of Higher-Order Cryptography}\label{sect:impossibility}
In this section, we will give both negative and positive results about the possibility
of defining a deterministic polytime strategy for the game
$\Strpar{\idpol}\arrow !_{q}(\Strpar{r}\arrow\BoolGame)\arrow\Strpar{p}$
which could be used to authenticate functions. When $r$ is linear, this is provably impossible, as
proved in Section~\ref{sect:compression} below. When, instead, $r$ is logarithmic (and
$q$ is at least linear), a positive result can be given, see Section~\ref{sect:positive}.

\begin{wrapfigure}{r}{.51\textwidth}
  \fbox{
    \begin{minipage}{.495\textwidth}
      \begin{center}
        \begin{tikzpicture}[scale=0.57]
          \node (t1) at (1,1){$!_{q}(\Strpar{r}$};
          \node (t2) at (3,1){$\arrow$};
          \node (t3) at (5,1){$\BoolGame)$};
          \node (t4) at (7,1){$\arrow$};
          \node (t5) at (9,1){$\Strpar{p}$};
          \node (o1) at (-1,0){O};
          \node (mo1) at (9,0){$?_{\Strpar{p}}$};
          \node (p1) at (-1,-1){P};
          \node (mp1) at (5,-1){$(1,?_{\BoolGame})$};
          \node (o2) at (-1,-2){O};
          \node (mo2) at (1,-2){$(1,?_{\Strpar{r}})$};
          \node (p3) at (-1,-3){P};
          \node (mo2) at (1,-3){$(1,s_1)$};
          \node (o4) at (-1,-4){O};
          \node (mo4) at (5,-4){$(1,t_1)$};
          \node (p) at (3,-5){$\vdots$};
          \node (p6) at (-1,-6){P};
          \node (mp6) at (5,-6){$(m,?_{\BoolGame})$};
          \node (o7) at (-1,-7){O};
          \node (mo7) at (1,-7){$(m,?_{\Strpar{r}})$};
          \node (p8) at (-1,-8){P};
          \node (mo2) at (1,-8){$(m,s_m)$};
          \node (o9) at (-1,-9){O};
          \node (mo9) at (5,-9){$(m,t_m)$};
          \node (p10) at (-1,-10){P};
          \node (mp10) at (9,-10){$v$};    
        \end{tikzpicture}
      \end{center}
  \end{minipage}}
  \caption{Plays (of maximal length) for the game $!_{q}(\Strpar{r}\arrow\BoolGame)\arrow\Strpar{p}$ }\label{fig:playssof}
\end{wrapfigure}
But how would a strategy for the game
$!_{q}(\Strpar{r}\arrow\BoolGame)\arrow\Strpar{p}$ look like? Plays for
this game are as in Figure~\ref{fig:playssof}. A strategy for such
a game is required to determine the value of the query $s_{i+1}\in\PBinStr{r(n)}$ based
on $t_{1},\ldots,t_{i}\in\Bool$. Moreover, based on
$t_1,\ldots,t_m$ (where $m\leq q(n)$), the strategy should be able to produce the value
$v\in\PBinStr{p(n)}$. Strictly speaking, the strategy should also be able to respond
to a situation in which the opponent directly replies
to a move $(i,?_{\BoolGame})$ by way of a truth value $(i,t_i)$, without
querying the argument. This is however a signal that the agent
with which the strategy is interacting represents a \emph{constant} function,
and we will not consider it in the following.

The way we will prove authentication impossible when $r$ is linear consists
in showing that since $q$ is polynomially bounded (thus negligibly smaller
than the number of possible queries of type $\Strpar{r}$ any function is allowed
to make to its argument), there are many argument functions
$\PBinStr{r(n)}\rightarrow\Bool$ which are indistinguishable, and would thus receive
the same tag. In Section~\ref{sect:compression}, we prove
that the (relatively few) influential coordinates can be efficiently determined.
\subsection{Efficiently Determining Influential Variables}\label{sect:influential}
A key step towards proving our negative result comes from the theory
of influential variables in decision trees. In this section, we are going
to give some preliminary definitions about it, without any aim at being
comprehensive (see, e.g., \cite{ODSS05}).

The theory of influential variables is concerned with boolean functions,
and with their statistical properties. The first definitions we need, then,
are those of variance and influence (from now on, metavariables like
$\parone,\partwo,\parthree$ stand for natural number unrelated to
the security parameter, unless otherwise specified). Given a natural number $N\in\NN$, $[N]$ denotes the set
$\{1,\ldots,N\}$.  Whenever $j\in[\parone]$,
$e_j\in\PBinStr\parone$ is the sequence which is everywhere $0$ except
on the $j$-th component, in which it is $1$.
\begin{definition}[Variance and Influence]
For every distribution $\distrone$ over $\PBinStr{\parone}$, and
$\bfunone:\PBinStr\parone \rightarrow\Bool$, we write
$\Var{\distrone}{\bfunone}$ for the value
$\ev{\bfunone(\distrone)^2} - \ev{\bfunone(\distrone)}^2= \Pr_{x,y\sim D}(\bfunone(x) \neq \bfunone(y))$, called the \emph{variance of $\bfunone$ under $\distrone$.}
For every distribution $\distrone$ over $\PBinStr\parone$,
$\bfunone:\PBinStr\parone \rightarrow\Bool$, and $j \in [N]$, we
define the \emph{influence of $j$ on $\bfunone$ under $\distrone$},
written $\Infl{\distrone}{j}{\bfunone}$, as
$\Pr_{x \sim D}[\bfunone(x) \neq \bfunone(x \oplus e_j)]$.
\end{definition}
The quantity $\Infl{\distrone}{j}{\bfunone}$ measures how much,
on the average, changing the $j$-th input to $\bfunone$ influences
its output. If $\bfunone$ does not depend too much on the $j$-th input,
then $\Infl{\distrone}{j}{\bfunone}$ is close to $0$, while it is
close to $1$ when switching the $j$-th input has a strong effect on the
output.
\begin{example}\label{ex:parityfun}
  Let $\mathit{PARITY}_\parone:\PBinStr\parone\rightarrow\Bool$ be the parity function
  on $\parone$ bits. It holds that
\begin{align*}
  \Infl{\distrone}{j}{\mathit{PARITY}_\parone}&=
  \Pr_{x \sim \distrone}[\mathit{PARITY}_\parone(x) \neq \mathit{PARITY}_\parone(x \oplus e_j)]\\
  &=\sum_{x}\distrone(x)\cdot |\mathit{PARITY}_\parone(x)-\mathit{PARITY}_\parone(x \oplus e_j)|=
  \sum_{x}\distrone(x)=1
\end{align*}
\end{example}
If $\bfunone:A \rightarrow \PBinStr L$, and $t\in [L]$, we define
$\bfunone_t:A \rightarrow \Bool$ to be the function that on input $x\in
A$ outputs the $t$-th bit of $\bfunone(x)$.

The kind of distributions over $\PBinStr\parone$ we will be mainly
interested are the so-called \emph{semi-uniform} ones, namely those in which
 some of the $\parone$ bits have a fixed value, while the others
take all possible values with equal probability. It is thus natural
to deal with them by way of partial functions.
  For every partial function $\funone:[N]
\rightharpoonup \Bool$ we define $\dom{\funone}\subseteq [N]$ to be
the set of inputs on which $\funone$ is defined, and
$\unif{\funone}$ to be the uniform distribution of $x$ over $\PBinStr N$
\emph{conditioned on} $x_j = g(j)$ for every $j \in\dom{g}$, i.e., the
distribution defined as follows:
$$
\unif{\funone}(x)=
\left\{
\begin{array}{ll}
  \frac{1}{2^{N-|\dom{g}|}} & \mbox{if $x_j = g(j)$ for every $j \in\dom{g}$};\\
  0 & \mbox{otherwise}.
\end{array}
\right.
$$
If a distribution $\distrone$ can be written as $\unif{g}$, where
$\funone:[N] \rightharpoonup \Bool$, we say that $\distrone$ is an
\emph{$\dom{\funone}$-distribution}, or a \emph{semi-uniform}
distribution. 
Given a distribution $\distrone:\PBinStr\parone\rightarrow\RR_{[0,1]}$,
some index $j\in[\parone]$ and a bit $b\in\Bool$,
the expression $\subst{\distrone}{j}{b}$ stands for the conditioning of
$\distrone$ to the fact that the $j$-th boolean argument is $b$. Note
that if $\distrone$ is an $\nsetone$-distribution and $j\in[\parone]\setminus\nsetone$,
then $\subst{\distrone}{j}{b}$ is an $\nsetone\cup\{j\}$-distribution.

A crucial concept in the following is that of a decision tree, which
is a model of computation for boolean functions in which the
interactive aspects are put upfront, while the merely computational
aspects are somehow abstracted away.
\begin{definition}[Decision Tree]
Given a function $F$, a \emph{decision tree $T$ for $R$} is a finite
ordered binary tree:
\begin{itemize}
\item
  Whose internal nodes are labelled with an index $i\in[\parone]$.
\item
  Whose leaves are labelled with a bit $b\in\Bool$.
\item
  Such that whenever a path ending in a leaf labelled with $b$
  is consistent with $x\in\PBinStr\parone$, it holds that
  $F(x)=b$.
\end{itemize}
The \emph{depth} of any decision tree $T$ is defined
as for any tree.
\end{definition}
\begin{example}
  An example of a decision tree that computes
  the function $\mathit{PARITY}_3:\PBinStr 3\rightarrow\Bool$ defined
  in Example~\ref{ex:parityfun} can be found in Figure~\ref{fig:decisiontree}.
\end{example}
\begin{wrapfigure}{r}{.51\textwidth}
  \fbox{
    \begin{minipage}{.495\textwidth}
      \begin{center}
        \includegraphics{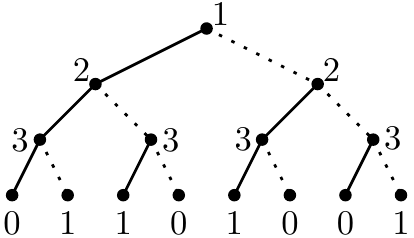}
      \end{center}
  \end{minipage}}
  \caption{A decision tree for $\mathit{PARITY}_3$}\label{fig:decisiontree}
\end{wrapfigure}
The following result, which is an easy corollary of some well-known
result in the literature, allows to put the variance and the influence
in relation whenever the underlying function can be computed by way of
a decision tree of limited depth.
\begin{lemma} \label{lem:infvariable}
  Suppose that $F$ is computable by a decision tree of depth at most
  $q$ and $g:[N] \rightarrow\Bool$ is a partial function. Then there
  exists $j\in [N] \setminus \dom{g}$ such that
  $$
  \Infl{\unif{g}}{j}{\bfunone} \geq \tfrac{\Var{\unif{g}}{\bfunone}}{q}
  $$
\end{lemma}
\begin{proof}
  This is basically corollary of Corollary 1.2 from~\cite{ODSS05}.
\end{proof}

Every decision tree $\dtone$ makes on any input a certain number of
queries, which of course can be different for different inputs.
If $\distrone$ is a distribution, $\nsetone$ is a subset of $[N]$ and $\dtone$ is a decision tree,
we define $\enq{\distrone}{\nsetone}{\dtone}$ as the expectation over $x \sim \distrone$ of
the number of queries that $\dtone$ makes on input $x$ outside of $\nsetone$, which
is said to be \emph{the average query complexity of $\dtone$ on $\distrone$ and $\nsetone$}. The
following result relates the query complexity before and after
the underlying semi-uniform distribution is updated: if we fix the value
of a variable, then the average interactive complexity goes
down (on the average) by at least the influence of the variable:
\begin{lemma} \label{lem:reducedelta}
  For every decision tree $\dtone$ computing a function $\bfunone$, $\nsetone\subseteq
  [\parone]$, $j \in [\parone] \setminus \nsetone$, and $\nsetone$-distribution $\distrone$,
  it holds that
  $$
  \tfrac{1}{2}\enq{\subst{\distrone}{j}{0}}{\nsetone\cup\{j\}}{\dtone}+
  \tfrac{1}{2}\enq{\subst{\distrone}{j}{1}}{\nsetone\cup\{j\}}{\dtone}
  \leq \enq{\distrone}{\nsetone}{\dtone}
  - \Infl{\distrone}{j}{\bfunone}
  $$
\end{lemma}
\longv{
  \begin{proof}
    We start by showing that 
    \begin{equation}\label{equ:intermediate}
      \Delta_{D,S \cup\{j\}}(T) \leq \Delta_{D,S}(T) - \Infl{\distrone}{j}{\bfunone}
    \end{equation}
    For every $L \in [\parone]$, let $I_L$ be the 
    random variable that is equal to one if $L$ is queried by $\dtone$
    on $x\sim D$ and equal to $0$ otherwise.  Therefore we can rewrite
    (\ref{equ:intermediate}) as
    $$
    \sum_{L \not\in S \cup \{j \} } \ev{I_L} \leq \sum_{L \not\in S} \ev{I_L} - \Infl{\distrone}{j}{\bfunone}
    $$
    or
    \begin{equation}\label{equ:final}
    \ev{I_j} \geq \Infl{\distrone}{j}{\bfunone} \;.
    \end{equation}
    But for every $x$ on which $j$
    is not queried by the tree $\dtone$, $\dtone(x) = \dtone(x\oplus e_j)$. Hence if
    $j$ is influential w.r.t. $x$ then certainly $j$ is queried, and
    (\ref{equ:final}) follows.
    For obvious reasons $\distrone = \tfrac{1}{2}\subst{\distrone}{j}{0} +
    \tfrac{1}{2}\subst{\distrone}{j}{1}$ and so in particular
    $$
    \enq{\distrone}{\nsetone\cup \{j \}}{\dtone} =
    \tfrac{1}{2}\enq{\subst{\distrone}{j}{0}}{\nsetone\cup \{j \}}{\dtone}+ 
    \tfrac{1}{2}\enq{\subst{\distrone}{j}{1}}{\nsetone\cup \{j \}}{\dtone}
    $$
    from which the thesis easily follows, given (\ref{equ:intermediate}).
  \end{proof}
  }
By somehow iterating over Lemma \ref{lem:reducedelta}, we can get the following
result, which states that fixing enough coordinates, the variance can
be made arbitrarily low, and that those coordinates can be efficiently
determined:
\begin{theorem} \label{thm:mainuncompressed}
    For every $\bfunone:\PBinStr\parone \rightarrow \PBinStr L$ such that for
    every $t\in [L]$, $F_t$ is computable by a decision tree of
    depth $\leq q$, and every $\varepsilon >0$, there exist a natural
    number $m \leq L
    q^2/\varepsilon$ and a partial function $g:[N] \rightarrow\Bool$
    where $|\dom{g}| \leq m$ such that
    \longv{\begin{equation}\label{eq:smallvariance}}
    \shortv{$}
    \Var{\unif{g}}{F_t} \leq \varepsilon
    \shortv{$}
    \longv{\end{equation}}
    for every $t \in \{1,\ldots,L\}$.  Moreover,    
    there is an polytime randomized algorithm $\mathsf{A}$
    that on input $\bfunone$, $\delta>0$, and $\varepsilon>0$,
    makes at most $O(L N)\cdot poly(q/(\delta\varepsilon))$ queries to $F$
    and outputs such a partial function $g$ with $|\dom{g}| \leq O((L
    q^2)/(\varepsilon\delta))$ with
    probability at least $1-\delta$.
  \end{theorem}
  \longv{
\begin{proof}
  For every $t\in [L]$ let $T_t$ be a tree computing $F_t$ and
  having depth at most $q$.
  The algorithm $\mathbf{A}$ works as follows:
  \begin{center}
  \textbf{Algorithm $\mathbf{A}(\bfunone,\varepsilon,\delta):$}
  \begin{enumerate}
  \item
    $i\leftarrow 0$;
  \item
    $g_0\leftarrow \emptyset$;
  \item
    $D_0\leftarrow \unif{g_0}$;
  \item
    $\gamma\leftarrow\delta/(NL q^2/(\delta \varepsilon))$
  \item\label{instr:firstloop}
    For every $t\in [L]$, perform $v_t\leftarrow\mathbf{EstimateVar}(D_i,F_t,\varepsilon/3,\gamma)$;
  \item
    If for every $t\in [L]$, it holds that $v_t \leq (2/3)\varepsilon$, then return $g_i$, otherwise continue.
  \item
    Let $t_{\uparrow}$ be a $t$ such that $v_t \geq 2\varepsilon/3$;
  \item
    For every $j\in [N]$, perform $n_j\leftarrow\mathbf{EstimateInf}(D_i,F_{t_{\uparrow}},\varepsilon/(10q),\gamma)$;
  \item
    Let $j$ be such that $n_j\geq 0.2 \cdot \varepsilon/q$;
  \item
    $b\leftarrow\{0,1\}$;
  \item
    $g_{i+1}\leftarrow g_i\cup\{(j,b)\}$;
  \item
    $D_{i+1}\leftarrow \unif{g_{i+1}}$;
  \item
    $i\leftarrow i+1$.
  \item
    Go back to \ref{instr:firstloop}.
  \end{enumerate}
  \end{center}




   
  The algorithm can be studied in two steps:
  \begin{itemize}
  \item
    We first of all analyse the complexity of the algorithm $\mathbf{A}$.
    An analysis of the runtime gives us the following result:
    \begin{proposition} \label{prop:expectednumberofiter}
      The number of times Algorithm $\mathbf{A}$ executes its main loop
      is, in expectation, at most $O(L q^2/\varepsilon)$.
    \end{proposition}
    \begin{proof}
    We define the potential function
    $$
    \varphi(i) = \sum_{t=1}^L \Delta_{D_i,s(g_i)}(T_t)
    $$ 
    where $T_t$ is the tree computing $F_t$.  The algorithm does not
    of course have access to this tree but this potential function is
    only used in the analysis. Initially, we are guaranteed that    
    $\varphi(0) \leq L q$ and by definition $\varphi(i) \geq 0$ for
    every $i$. Hence to show that we end in expected number of steps
    $O(L q^2/\varepsilon)$ it is enough to show that in expectation
    $$
    \mathbb{E}[\varphi(i) - \varphi(i+1)] \geq \Omega(\varepsilon/q)
    $$
    for every step $i$ in which we do not stop.
    Indeed, let $i$ be such a step and let $t_0$ be the index such
    that our estimate for $Var(F_{t_0}(D_i)) \geq (2/3)\varepsilon$
    and $j_0$ the influential variable for $F_{t_0}$ dud we
    find. Recall that $g_{i+1}$ is obtained by extending $g_i$ to
    satisfy $g_{i+1}(j_0) = b$ for a random $b \in \Bool$.  For
    every $t \neq t_0$, the expectation over $b$ of
    $\Delta_{D_{i+1},s(g_{i+1})}(T_t)$ is equal to the expected number
    of queries that $T_t$ makes on $D_i$ that do not touch the set
    $s(g_i) \cup \{j \}$.  This number can only be smaller than
    $\Delta_{D_i,s(g_i)}(T_t)$.  For $t=t_0$, by
    Lemma~\ref{lem:reducedelta}, $\mathbb{E}[
      \Delta_{D_{i+1},s(g_{i+1})}(F_t)] \leq \mathbb{E}[
      \Delta_{D_{i+1},s(g_{i+1})}(F_t)] + \Omega(\varepsilon/q)$.
    Hence, in expectation $\varphi(i+1)$ is smaller than $\varphi(i)$
    by an additive factor of $\Omega(\varepsilon/q)$ which is what we
    wanted to prove.
    \end{proof}
  \item
  A further discussion is needed to ensure that the desired level
  of accuracy is actually attained by the algorithm. The reason
  why $\gamma$ is set to $\delta/T$ (where $T=
  O(NL q^2/(\delta \varepsilon))$ is our bound on the number
  of times we need to use these estimates) is to ensure $\gamma$
  to be small enough so that we can take a
  union bound over the chances that any of our estimates
  fail.  Thus the algorithm will
  only stop when it reach $D_i = U_{g_i}$ for which the variance of
  $F_t$ is smaller than $\varepsilon$ for every $t\in [L]$.
\item
  Once we have Proposition~\ref{prop:expectednumberofiter}, the
  theorem follows by noting that if the expected number of iterations
  of the main loop is $m$, the probability that we make more than
  $m/\delta$ iterations is at most $\delta$.  (Since the potential
  function undergoes a biased random walk, a more refined analysis can
  be used to show that $O(m\log(1/\delta))$ iterations will suffice,
  but this does not matter much for our final application and so we
  use the cruder bound of $m/\delta$.)
  \end{itemize}
  \end{proof}}
\subsection{On the Impossibility of Authenticating Functions}\label{sect:compression}
Theorem \ref{thm:mainuncompressed} tells that for every first-order
boolean function which can be computed by a decision tree of
low depth, there exist relatively few of its coordinates that,
once fixed, determine the function's output with very high
probability. If $N$ is exponentially larger than $q$, in particular
there is no hope for such a function to be a secure message authentication code.
In this section, we aim at exploiting all this
to prove that higher-order authentication
is not possible if the argument function has too large a domain.
In order to do it, we build a third-order randomized algorithm,
which will be shown to fit into the our game-theoretic framework.

More specifically, we are concerned with the cryptographic
properties of strategies for the parametrized game
$\sof{q}{r}{p}=!_{q}(\Strpar{r}\arrow\BoolGame)\arrow\Strpar{p}$
and, in particular, with the case in which $r$ is the identity $\idpol$, i.e.
we are considering the game $\linsof{q}{p}=\sof{q}{\idpol}{p}$.
Any such strategy, when deterministic, can be seen as computing 
a family of functions $\{F_n\}_{n\in\NN}$ where
$F_n: (\PBinStr{n} \rightarrow \Bool) \rightarrow \PBinStr{p(n)}$.
How could we fit all this into the hypotheses of Theorem~\ref{thm:mainuncompressed}?

The definitions of variance, influence, and decision tree
can be easily generalised to functions in the form
$F: (\PBinStr{\parone} \rightarrow \Bool) \rightarrow \PBinStr{\partwo}$.
Of course the underlying distribution $\distrone$ must be
a distribution over functions $\PBinStr{\parone} \rightarrow \Bool$.
The parameter $N$ can be fixed in such a way that $n<N<2^n$, where
$n$ is the security parameter. For simplicity we will
choose $N$ to be a power of $2$, which hence
divides $2^n$. \longv{In other words, $N=2^p$, where $p$ is strictly between
$\log_2(n)$ and $n$. Eventually we will set $N$ to be some large enough
polynomial in $n$, i.e., $p$ will be in the form $a\cdot\log_2(n)+b$.}
\begin{definition}[Extension]
  For every $x \in \PBinStr{N}$, we define the \emph{extension of $x$},
  denoted by $f_x$ as the function $f_x:\PBinStr{n} \rightarrow \Bool$
  such that for every $i\in [2^n]$ (identifying $\PBinStr{n}$ with the
  numbers $\{0,\ldots, 2^n-1 \}$ in the natural way),
  it holds that $f_x(i) = x_{\lfloor i/N \rfloor + 1}$.
  That is, $f_x$ is the function that outputs $x_1$ on the first $N$
  inputs, outputs $x_2$ on the second $N$ inputs, and so on and so
  forth. \longv{If $F:(\PBinStr{n} \rightarrow \Bool) \rightarrow \Bool$ is
  a function, then we define $\tilde{F}:\PBinStr{N} \rightarrow \Bool$
  as follows: for every $x\in \PBinStr{N}$, $\tilde{F}(x) = F(f_x)$.}
  Given a distribution $\distrone$ over $\PBinStr{N}$,
  a distribution over functions $\PBinStr{n} \rightarrow \Bool$
  can be formed in the natural way as $f_{\distrone}$.
\end{definition}

We will also make use of the following slight variation on
the classic notion of Hamming distance: define $H(\cdot,\cdot)$ to be the
so-called \emph{normalized Hamming distance}. In fact, we overload
the symbol $H$ and use it for both
\emph{strings} in $\PBinStr{N}$ and \emph{functions} in $\PBinStr{n} \rightarrow X$ for
some set $X$. That is, if $x,y\in \{0,1\}^N$ then
$H(x,y) = \Pr_{j\in [N]}[x_j \neq y_j]$
while if $f,g\in \PBinStr{n} \rightarrow X$ then
$H(f,g) = \Pr_{i\in \PBinStr{n}}[f(i) \neq g(i)]$.
\longv{
We use the following simple lemma
\begin{lemma}
For every $x,y\in \PBinStr{N}$, $H(f_x,f_{y}) = H(x,y)$
\end{lemma}
\begin{proof}
  If we choose $i\in \PBinStr{n}$ then (identifying $\PBinStr{n}$ with $\{
  0,\ldots 2^n-1\}$), since $N$ divides $2^n$, the distribution
  $\lfloor i/N \rfloor$ is uniform over $\{0,\ldots, N-1 \}$ and so
  the two probabilities in the definition of the Hamming distance are
  the same. More formally:
  \begin{align*}
  H(f_x,f_y)&= \Pr_{i\sim\PBinStr{n}}[f_x(i) \neq f_y(i)]=\sum_{i\in\PBinStr{n}}\frac{1}{2^n}\left|f_x(i)-f_y(i)\right|\\
  &=\sum_{i\in\PBinStr{n}}\frac{1}{2^n}\left|x_{\lfloor i/N \rfloor + 1}-y_{\lfloor i/N \rfloor + 1}\right|
  =2^{n-p}\sum_{j\in[N]}\frac{1}{2^n}\left|x_{j}-y_{j}\right|\\
  &=\sum_{j\in[N]}\frac{1}{2^p}\left|x_{j}-y_{j}\right|
  =\Pr_{j\sim [N]}[x_j \neq y_j]=H(x,y).
  \end{align*}
\end{proof}}
The game $\TStrpar{q}$ is a slight variation on
$\Strpar{q}$ in which the returned string is in
a ternary alphabet $\{0,1,\bot\}$. Any strategy
for $\TStrpar{q}$ can be seen as representing a
(family of) partial functions from $[q(n)]$ to
$\Bool$.
\begin{theorem} \label{thm:dropvariance}
  For every $\varepsilon,\delta>0$,
  there is a polytime probabilistic strategy $\strinfvar_{\varepsilon,\delta}$ on the game
  $!_{s}(\linsof{q}{p})\arrow\TStrpar{t}$ such that
  for every deterministic strategy $\stratone$
  on $\linsof{q}{p}$ computing $\{F_n\}_{n\in\NN}$,
  the composition $\stratone;\strinfvar_{\varepsilon,\delta}$,
  with probability at least $1-\delta$, computes
  some functions $g_n:[t(n)] \rightarrow \Bool$
  such that $\Var{f_{\unif{g_n}}}{F_n} \leq \varepsilon$ and $|\dom{g_n}| \leq \mathcal{O}(p(n)
  q^2(n)/(\delta \epsilon))$. 
\end{theorem}
\longv{
\begin{proof}
  Please observe that $\tilde{F}$ is a \emph{function} from
  $\PBinStr{N}$ to $\PBinStr{L}$ which, since $F$ is
  computable by a decision tree of depth $q$, can be computed
  itself by a decision tree of the same depth.
  The result follows by just looking at the function $\tilde{F}$ and
  applying to it Theorem~\ref{thm:mainuncompressed}. Actually,
  the queries that algorithm $\mathbf{A}$ makes to $\tilde{F}$
  correctly translate into queries of the form $f_i=f_{x_i}$.
\end{proof}}
Remarkably, the strategy $\strinfvar_{\varepsilon,\delta}$ is
that it infers the ``influential variables'' of the strategy with which it interacts
\emph{without} looking at how the latter queries its argument function, something
which would anyway be available in the history of the interaction.

We can now obtain the main result of this section.
\begin{theorem} \label{thm:main}
  For every $\delta$ there is a polytime probabilistic
  strategy $\strcoll_{\delta}$ on a game
  $!_{s}(\linsof{q}{p})\arrow(\Strpar{\idpol}\arrow\BoolGame)\tensor(\Strpar{\idpol}\arrow\BoolGame)$
  such that for every deterministic strategy $\stratone$ on $\linsof{q}{p}$
  computing $\{\bfunone_n\}_{n\in\NN}$, the composition
  $\stratone;\strcoll_{\delta}$, with probability at least $1-\delta$,
  computes two function families $g,h$ with
  $g_n,h_n:\PBinStr{n}\rightarrow\Bool$ such
  that
  \begin{enumerate}
  \item
    $H(g_n,h_n) \geq 0.1$
  \item
    $F_n(g_n)=F_n(h_n)$.
  \item
    For every function $f$ on which
    $\strcoll_{\delta}$ queries its argument, it holds that
    $H(f,g_n) \geq 0.1$ and $H(f,h_n) \geq 0.1$.
  \end{enumerate}
  
\end{theorem}
This shows that $\strcoll$ finds a \emph{collision} for $F_n$ as
a pair of functions that are different from each other (and in fact
significantly different in Hamming distance) but for which $F_n$ outputs
the same value, and hence $F$ cannot be a collision-resistant hash
function.  Moreover, because the functions are far from those queried,
this means that $F_n$ cannot be a secure message authentication code either,
since by querying $F_n$ on $g_n$, the adversary can predict the value of
the tag on $h_n$.

\longv{
\begin{proof}
We run the algorithm of Theorem~\ref{thm:dropvariance} and obtain a
partial function $g$ such that $|S(g)| \leq 10 L q^2 /
(\delta^2\epsilon)$ and such that the variance of every one of the
functions $F_1,F_2,\ldots,F_L$ is smaller than $\delta\epsilon$.
We choose $N$ to be large enough so that the bound $10 L q^2/
(\delta^2 \epsilon)$ on $|S(g)|$ is smaller than $\delta N / 100$.
Once we achieve this, we sample $x'$ and $x''$ independently from
$U_g$ and set $g= f_{x'}$ and $h = f_{x''}$. Since we fixed only a
$0.01 \delta$ fraction of the coordinates, and $x'$ and $x''$ are
random over the remaining ones, using the standard Chernoff bounds
with high probability $g$ and $h$ will differ on at least $1 -
2^{-L} - \delta/10$ fraction of the remaining coordinates from each
other and also from all other previous queries. (Specifically, the
probability of the difference being smaller than this is exponentially
small in $\delta N$ and we can make $N$ big enough so that this is
much smaller than $\delta/(L m)$ and so we can take a union bound
on all queries.)  On the other hand, because of the variance, the
probability that $F_i(g) \neq F_i(h)$ for every $i$ is less than
$\delta/L$ and so we can take a union bound over all $L$ $i$'s
to conclude that $G(g)=G(h)$ with probability at least $1-\delta$.
\end{proof}}
\subsection{A Positive Result on Higher-Order Pseudorandomness}\label{sect:positive}
We conclude this paper by giving a positive
result. More specifically, we prove that pseudorandomness
can indeed be attained at second order, but at an high price, namely
by switching to the type
$\logsof{p}=\sof{\idpol}{\ulg}{p}$.
This indeed has the same structure of $\linsof{q}{p}$, but the argument
function takes in input strings of \emph{logarithmic size} rather
than linear size. Moreover, the argument function
can be accessed a linear number of times, which is enough
to query it on \emph{every possible} coordinate.

The fact that that a strategy on $\logsof{p}$ can query its argument
on every possible coordinate renders the attacks described in the
previous section unfeasible. Actually, $\logsof{p}$ can be seen as
an \emph{interactive} variation of the
game $\Strpar{\idpol}\arrow\Strpar{p}$, for
which pseudorandomness is well known to be attainable starting from
standard cryptographic assumptions~\cite{GoldreichI06}: simply, instead of taking in
input \emph{the whole} string \emph{at once}, it queries it \emph{bit
  by bit}, in a certain order. A \emph{random strategy} of that type,
then, would be one that, using the notation from Figure~\ref{fig:playssof},
\begin{itemize}
\item
  Given $t_{1},\ldots,t_{i}\in\Bool$, returns a
  string $s_{i+1}$ uniformly chosen at random
  from $\PBinStr{r(n)}-\{s_1,\ldots,s_i\}$,
  this for every $i<q(n)$.
\item
  Moreover, based on $t_1,\ldots,t_{r(n)}$, the strategy
  produces a string $v$ chosen uniformly at random
  from $\PBinStr{p(n)}$.
\end{itemize}
Please notice that this random strategy can be considered as a
\emph{random functionals} in $(\PBinStr{r(n)}\rightarrow\Bool)\rightarrow\PBinStr{p(n)}$
only if $r(n)$ is logarithmic, because this
way the final result $v$ is allowed to depend on the value of the
input function in \emph{all possible coordinates}. The process of
generating such a random strategy uniformly can be seen\footnote{the
  strategy at hand would, strictly speaking, need exponentially many
  random bits, which are not allowed in our model.}
as a probabilistic strategy $\strarandsof$.

We are now ready to formally define pseudorandom functions:
\begin{definition}[Second-Order Pseudorandom Function]
  A deterministic polytime strategy $\stratone$ on
  $\Strpar{\idpol}\arrow\logsof{p}$ is said to be pseudorandom iff for every
  probabilistic polytime strategy $\mathcal{A}$
  on $!_{s}\logsof{p}\arrow\BoolGame$ there is
  a negligible function $\varepsilon:\NN\rightarrow\RR_{[0,1]}$
  such that
  $$
  |\Pr(\conv{!_{s}(\stracat;\stratone);\mathcal{A}}{n})-\Pr(\conv{!_{s}\strarandsof;\mathcal{A}}{n})|\leq\varepsilon(n)
  $$
\end{definition}

\newcommand{\ftos}{\mathsf{first2second}}
The way we will build a pseudorandom function consists in constructing
a \emph{deterministic} polytime strategy $\ftos$ for the following game
$!_{\idpol+1}(\LTStrpar{\idpol}\arrow\Strpar{w})\arrow\logsof{p}$, where
$w\in\pols$ is such that $w\geq\ulg$ and $w\geq p$.
The way the strategy is defined is in Figure~\ref{fig:fromfirsttosecond}.
\begin{wrapfigure}{r}{.62\textwidth}
  \fbox{
    \begin{minipage}{.6\textwidth}
      \begin{center}
        \scriptsize
        \begin{tikzpicture}[scale=0.37]
          \node (t1) at (1,1){$!_{\idpol+1}(\LTStrpar{\idpol}$};
          \node (t2) at (3,1){$\arrow$};
          \node (t3) at (5,1){$\Strpar{w})$};
          \node (t4) at (7,1){$\arrow$};
          \node (t5) at (9,1){$!_{\idpol}(\Strpar{\ulg}$};
          \node (t4) at (11,1){$\arrow$};
          \node (t5) at (13,1){$\BoolGame)$};
          \node (t4) at (15,1){$\arrow$};
          \node (t5) at (17,1){$\Strpar{p})$};
          \node (o1) at (-4,0){O};
          \node (mo1) at (17,0){$?_{\Strpar{p}}$};
          \node (p) at (3,-1){$\vdots$};          
          \node (p1) at (-4,-2){P};
          \node (mp1) at (5,-2){$(i,?_{\Strpar{w}})$};
          \node (o2) at (-4,-3){O};
          \node (mo2) at (1,-3){$(i,?_{\LTStrpar{\idpol}})$};
          \node (p3) at (-4,-4){P};
          \node (mo2) at (1,-4){$(i,(z_1\cdots z_{i-1})$};
          \node (o4) at (-4,-5){O};
          \node (mo4) at (5,-5){$(i,j_i))$};
          \node (p1) at (-4,-6){P};
          \node (mp1) at (13,-6){$(i,?_{\BoolGame})$};
          \node (o2) at (-4,-7){O};
          \node (mo2) at (9,-7){$(i,?_{\Strpar{\ulg}})$};
          \node (p3) at (-4,-8){P};
          \node (mo2) at (9,-8){$(i,\alpha(j_i,\{j_1,\ldots,j_{i-1}\})$};
          \node (o4) at (-4,-9){O};
          \node (mo4) at (13,-9){$(i,z_i)$};
          \node (p) at (3,-10){$\vdots$};
          \node (p1) at (-4,-11){P};
          \node (mp1) at (5,-11){$(n,?_{\Strpar{w}})$};
          \node (o2) at (-4,-12){O};
          \node (mo2) at (1,-12){$(n,?_{\LTStrpar{\idpol}})$};
          \node (p3) at (-4,-13){P};
          \node (mo2) at (1,-13){$(n,(z_1\cdots z_n))$};
          \node (o4) at (-4,-14){O};
          \node (mo4) at (5,-14){$(n,v)$};          
          \node (p10) at (-4,-15){P};
          \node (mp10) at (17,-15){$\beta(v)$};    
        \end{tikzpicture}
      \end{center}
  \end{minipage}}
  \caption{Plays in $\caracplay{\ftos}$}\label{fig:fromfirsttosecond}
\end{wrapfigure}
The function $\alpha$ interprets its first argument
(a string in $\PBinStr{w(n)}$), as an element of $\PBinStr{\ulg(n)}$
different from those it takes as second argument, and
distributing the probabilities uniformly. The function $\beta$,
instead, possibly discards some bits of the input and produces
a possibly shorter string.

The way the strategy $\ftos$ is defined so as to be such that
$\stratone;\ftos$ is statistically very close to the random strategy
whenever $\stratone$ is chosen uniformly at random among the strategies
for the parametric game $\LTStrpar{\idpol}\arrow\LTStrpar{w}$. This
allows us to prove the following:
\begin{theorem}
  Let $F:\{0,1\}^n\times\{0,1\}^{\leq n}\rightarrow\{0,1\}^{w(n)}$
  be a pseudorandom function and let $\stratone_F$ be the deterministic
  polytime strategy for the game
  $\Strpar{\idpol}\arrow!_{\idpol+1}(\LTStrpar{\idpol}\arrow\Strpar{w})$ obtained
  from $F$. Then, $\stratone_F;\ftos$ is second-order pseudorandom.
\end{theorem}
\section{Related Work}
Game semantics and geometry of interaction are among
the best-studied program semantic frameworks (see, e.g.~\cite{AJM00,HO00}), and can also be
seen as computational models, given their operational
flavor. This is particularly apparent in the work on abstract
machines~\cite{CurienHerbelin07,FG13}, but also on the so-called geometry of synthesis~\cite{Ghica07}. In this
paper, we are particularly interested in the latter use of game semantics,
and take it as the underlying computational model. Our game model
is definitionally strongly inspired by Wolverson~\cite{Wolverson09}: the main novelty
with respect to the latter is the treatment of randomized strategies,
and the bounded exponential construction, which together allows us
to account for efficient randomized higher-order computation.

This is certainly not the first paper in which cryptography
is generalized to computational models beyond the one of
first-order functions. One should of course cite Canetti's
universally composable security~\cite{Canetti01}, but also Mitchell et al.'s
framework, the latter based on process algebras~\cite{MRST06}. None
of them deals with the security of higher-order functions, though.
\bibliographystyle{plain}
\bibliography{main}

\begin{thebibliography}{10}

\bibitem{AJM00}
Samson Abramsky, Radha Jagadeesan, and Pasquale Malacaria.
\newblock Full abstraction for {PCF}.
\newblock {\em Inf. Comput.}, 163(2):409--470, 2000.

\bibitem{BlumMicali84}
Manuel Blum and Silvio Micali.
\newblock How to generate cryptographically strong sequences of pseudo-random
  bits.
\newblock {\em {SIAM} J. Comput.}, 13(4):850--864, 1984.

\bibitem{Canetti01}
Ran Canetti.
\newblock Universally composable security: {A} new paradigm for cryptographic
  protocols.
\newblock In {\em Proc. of {FOCS} 2001}, pages 136--145, 2001.

\bibitem{CP18}
Pierre Clairambault and Hugo Paquet.
\newblock Fully abstract models of the probabilistic lambda-calculus.
\newblock In Dan~R. Ghica and Achim Jung, editors, {\em Proc. of {CSL} 2018},
  volume 119 of {\em LIPIcs}, pages 16:1--16:17, 2018.

\bibitem{CookKapron89}
Stephen~A. Cook and Bruce~M. Kapron.
\newblock Characterizations of the basic feasible functionals of finite type
  (extended abstract).
\newblock In {\em Proceedings of {FOCS} {1989}}, pages 154--159. {IEEE}
  Computer Society, 1989.

\bibitem{CurienHerbelin07}
Pierre{-}Louis Curien and Hugo Herbelin.
\newblock Abstract machines for dialogue games.
\newblock {\em CoRR}, abs/0706.2544, 2007.

\bibitem{DanosHarmer02}
Vincent Danos and Russell Harmer.
\newblock Probabilistic game semantics.
\newblock {\em ACM Trans. Comput. Log.}, 3(3):359--382, 2002.

\bibitem{Feree17}
Hugo F{\'{e}}r{\'{e}}e.
\newblock Game semantics approach to higher-order complexity.
\newblock {\em J. Comput. Syst. Sci.}, 87:1--15, 2017.

\bibitem{FG13}
Olle Fredriksson and Dan~R. Ghica.
\newblock Abstract machines for game semantics, revisited.
\newblock In {\em Proc. of {LICS} 2013}, pages 560--569, 2013.

\bibitem{Ghica07}
Dan~R. Ghica.
\newblock Geometry of synthesis: a structured approach to {VLSI} design.
\newblock In {\em Proceedings of the 34th {ACM} {SIGPLAN-SIGACT} Symposium on
  Principles of Programming Languages, {POPL} 2007, Nice, France, January
  17-19, 2007}, pages 363--375, 2007.

\bibitem{GoldreichI06}
Oded Goldreich.
\newblock {\em Foundations of Cryptography: Volume 1}.
\newblock Cambridge University Press, 2006.

\bibitem{GGM86}
Oded Goldreich, Shafi Goldwasser, and Silvio Micali.
\newblock How to construct random functions.
\newblock {\em J. {ACM}}, 33(4):792--807, 1986.

\bibitem{HO00}
J.~M.~E. Hyland and C.{-}H.~Luke Ong.
\newblock On full abstraction for {PCF:} i, ii, and {III}.
\newblock {\em Inf. Comput.}, 163(2):285--408, 2000.

\bibitem{KatzLindell}
Jonathan Katz and Yehuda Lindell.
\newblock {\em Introduction to Modern Cryptography}.
\newblock Chapman \& Hall/CRC, 2007.

\bibitem{KawamuraCook12}
Akitoshi Kawamura and Stephen~A. Cook.
\newblock Complexity theory for operators in analysis.
\newblock {\em {TOCT}}, 4(2):5:1--5:24, 2012.

\bibitem{LongleyNormann}
John Longley and Dag Normann.
\newblock {\em Higher-Order Computability}.
\newblock Theory and Applications of Computability. Springer, 2015.

\bibitem{MRST06}
John~C. Mitchell, Ajith Ramanathan, Andre Scedrov, and Vanessa Teague.
\newblock A probabilistic polynomial-time process calculus for the analysis of
  cryptographic protocols.
\newblock {\em Theor. Comput. Sci.}, 353(1-3):118--164, 2006.

\bibitem{GoldreichII06}
Goldreich Oded.
\newblock {\em Foundations of Cryptography: Volume 2, Basic Applications}.
\newblock Cambridge University Press, 2009.

\bibitem{ODSS05}
Ryan O'Donnell, Michael~E. Saks, Oded Schramm, and Rocco~A. Servedio.
\newblock Every decision tree has an in.uential variable.
\newblock In {\em Proc. of {FOCS} 2005}, pages 31--39. {IEEE} Computer Society,
  2005.

\bibitem{Weihrauch13}
Klaus Weihrauch.
\newblock {\em Computable Analysis: An Introduction}.
\newblock Springer Publishing Company, Incorporated, 2013.

\bibitem{Wolverson09}
Nicholas Wolverson.
\newblock {\em Game semantics for an object-oriented language}.
\newblock PhD thesis, University of Edinburgh, {UK}, 2009.

\end{thebibliography}

\end{document}